\documentclass[12pt,reqno]{amsart}
\newtheorem{theorem}{Theorem}[section]
\newtheorem{definition}{Definition}[section]

\newtheorem{corollary}{Corollary}[section]
\newtheorem{proposition}{Proposition}[section]
\theoremstyle{definition}
\newtheorem{remark}{Remark}[section]

\newcommand{\be}{\begin{equation}}
\newcommand{\ee}{\end{equation}}
\newcommand{\bea}{\begin{eqnarray}}
\newcommand{\eea}{\end{eqnarray}}
\newcommand{\beb}{\begin{eqnarray*}}
\newcommand{\eeb}{\end{eqnarray*}}
\usepackage{amssymb,amsfonts,amsthm,setspace,indentfirst,mathrsfs}
\usepackage{minibox}
\usepackage{enumitem}
\usepackage{cite}

\numberwithin{equation}{section}
\begin{document}
%
\title{Curvature  Inheritance Symmetry on M-Projectively Flat Spacetimes}
\author[A. A. Shaikh, M. Ali, M. Salman and F. \"{O}. Zengin]{Absos Ali Shaikh$^{1}$, Musavvir Ali$^{2}$,	Mohammad Salman$^{3}$  and  F\"{u}sun \"{O}zen   Zengin$^{4}$}
\date{\today}
\address{\noindent\newline$^{1}$ Department of Mathematics,
	\newline University of Burdwan, 
	\newline Golapbag, Burdwan-713104,
	\newline West Bengal, India} 
\email{aask2003@yahoo.co.in$^1$, aashaikh@math.buruniv.ac.in$^1$}

\address{\noindent\newline$^{2,3}$ Department of Mathematics,
	\newline Aligarh Muslim University, 
	\newline Aligarh-202002, India} 
\email{musavvirali.maths@amu.ac.in$^2$}
\email{salman199114@gmail.com$^3$}

\address{\noindent\newline$^{4}$ Department of Mathematics,
	\newline Istanbul Technical University,  
	\newline 34469-Istanbul, Turkey.}
\email{fozen@itu.edu.tr$^4$}

%
%
\dedicatory{}
\begin{abstract}
	The paper aims to investigate curvature inheritance symmetry in M-projectively flat spacetimes. It is shown that the curvature inheritance symmetry in M-projectively flat spacetime is a conformal motion. We have proved that M- projective curvature tensor follows the symmetry inheritance property along a vector field $\xi$,  when spacetime  admits the conditions of both curvature inheritance symmetry  and conformal motion or motion along the vector field  $\xi$.  Also, we have derived some results for M-projectively flat spacetime with perfect fluid following the  Einstein field equations with a cosmological term and admitting  the curvature inheritance symmetry along the vector field $\xi$. We have shown that an M-projectively flat perfect fluid spacetime  obeying  the Einstein field equations with a cosmological term and admitting the curvature inheritance symmetry along a vector field $\xi$ is either a vacuum or satisfies the vacuum-like equation of state.	We have also shown that such spacetimes with the energy momentum tensor of an electromagnetic field distribution do not admit any curvature symmetry of general relativity. Finally, an example of M-projectively flat spacetime has been exhibited.
\end{abstract}
%
\maketitle
%

\section{\bf Introduction}
A pseudo-Riemannian manifold with a non-zero $ (0,2) $ type  Ricci tensor $R_{ij}$ is known as Einstein manifold if $R_{ij}$  is proportional to the metric tensor $  g_{ij}$.  In  other words, we can simply say that Einstein manifolds  form a natural subclass of the manifold of constant curvature \cite{Besse}.  In the field of general relativity \cite{Eishart} and  Riemannian geometry,  Einstein   manifolds  play a significant role. A spacetime is considered as a four dimensional connected  semi-Riemannian manifold $(V_4,  g)$ with Lorentzian metric $g$  of   signature $(-, +,+,+)$. Due to the casual character of vectors, the  Lorentzian manifolds \cite{Besse} are the  practical option for  the studying   general theory of relativity.

With  metric tensor $g_{ij}$ and  Riemannian connection $\nabla$,  let $(V_4, g)$ be a 4-dimensional differentiable manifold of class $C^\infty$.    M-projective curvature tensor of the manifold  $V_4$ was defined by  Pokhariyal and  Mishra \cite{GPPokh} in 1971 which can be written in the following form
\begin{equation}\label{1.1}
	\bar{W}^{h}_{ijk}= R^h_{ijk} + \frac{1}{6}(\delta^h_j R_{ik}-\delta^h_k R_{ij}+ R^h_j g_{ik}- R^h_k g_{ij}), 
\end{equation}
where 	$ R^h_{ijk}$ is the Riemannian curvature, $g_{ij}$ is the metric  and $ R_{ij}$ is  Ricci tensor in  $V_4$,  the tensor field $\bar{W}^{h}_{ijk}$ is  the M-projective curvature tensor. We note that the M-projective curvature tensor is a generalized curvature tensor, i.e., it satisfies the symmetry and skew-symmetry properties like Riemann curvature and satisfies Bianchi identity.  If $	\bar{W}^{h}_{ijk}=0$ on $V_4$ identically, then that manifold is known as M-projectively flat manifold. We mention that any M-projectively flat semi-Riemannian manifold is an Einstein manifold and a space of constant curvature. Many authors have studied about spacetimes and their geometries and physics such as: Abdussattar and Dwivedi \cite{Abdussattar and B Dwivedi} studied conharmonic symmetries in fluid spacetimes,  Zengin \cite{Zengin1} studied M-Projectively flat spacetimes, pseudo Z-symmetric spacetimes by Mantica and Suh \cite{Mantica1},  concircular curvature tensor and fluid spacetimes by Ahsan et al. \cite{Ahsan,ZAhsan1,ZAhsan2,ZAhsan3,AliM,Zahsan12} and many more. The properties of M-projective curvature tensor were studied by Ojha \cite{RH2,RH1} and M-projective curvature tensor on various structures has been studied by several authors (see, \cite{SKRH,Shyam,Shaikh2003,Shanmukha,Venkatesha,Zengin2}).

The present  paper is dedicated to certain investigations in general theory of  relativity for  M-projectively flat case. We have explored the condition when energy momentum tensor follow  symmetry inheritance property and its Lie derivative  vanishes in M-projectively flat spacetimes.

The Bianchi identities describe the interaction between the matter and free gravitational parts of the gravitational field, which is characterized by the  curvature tensor in general theory of relativity. In gravitational physics, the main objective of all investigations is the construction of the gravitational potential, which satisfies the Einstein field equations (EFEs)  with a  cosmological constant. Consider the  EFEs with cosmological term  as
\begin{equation}\label{1.2}
	R_{ij} - \frac{R}{2} g_{ij} + \wedge g_{ij} =  k ~  T_{ij},
\end{equation}

\noindent  where  $ g_{ij}$,  $R_{ij}$, $T _{ij}$, $\wedge$ and  $k$ are  respectively denote the  metric,  Ricci tensor, energy-momentum tensor,  cosmological term and  the gravitational constant. The   energy-momentum tensor   for a perfect fluid is defined by the following
\begin{equation}\label{1.3n}
	T_{ij}=(\mu + p) u_{i}u_{j}+ p g_{ij},
\end{equation} 
here  $\mu$,  the energy density, $p$  the isotropic pressure and $u^{i}$ is the velocity vector field   of the flow satisfying  $g_{ij} u^{j} = u_{i}$  for all  $i$, also $u_{i}u^{i}= -1$.

Geometrical symmetries of the spacetime are represented  by the following equation \cite{yanoK}
\begin{equation}\label{1.3}
	\pounds_\xi  A =2  \varOmega A
\end{equation}
where $ \pounds_\xi $ stands for the Lie derivative along the vector field $\xi^{i}$. The vector field $\xi^{i}$ may be timelike ($\xi^{i}\xi_{i}<0$), spacelike ($\xi^{i}\xi_{i}>0$) or null ($\xi^{i}\xi_{i}=0$),  `A'   denotes a geometrical/physical quantity and $\varOmega$  is a scalar function.

A simple example can be provided  as the metric inheritance symmetry in particular for $g_{ij} $, it is   conformal motion    (Conf M) \cite{Katzin} along a vector  field  $\xi^{i}$, so   
\begin{equation}\label{1.4}
	\pounds_\xi g_{ij}   =  2 \psi  g_{ij}, 
\end{equation} 

\noindent where $\psi$ is a conformal function,  equation \eqref{1.4} implies  motion/isometry for  $\psi =  0$. In this case $\xi$ is called a Killing vector \cite{yanoK}.   More than 30 geometric symmetries have been found in literature till  date. For detailed study of symmetry inheritance, see  \cite{Ahsan,klduggal2,Duggal,K.L.DuggalandR.Sharma,hall}.

In 1992, Duggal  \cite{Duggal} introduced the concept of curvature inheritance (CI) symmetry.  Curvature inheritance is the generalization of curvature collineation (CC),   which was  defined  by Katzin in 1969 \cite{Katzin}.  
\begin{definition}
	\cite{Duggal} A spacetime $(V_4, ~g)$ admits curvature inheritance symmetry along a 	 smooth vector field $\xi^{i}$,   if it satisfies  
	
	\begin{equation} \label{1.5N}
		\pounds_\xi R^h_{ijk} = 2 \varOmega R^h_{ijk}, 
	\end{equation}
	where $\varOmega = \varOmega(x^{i})$ is known as inheriting factor or inheritance function.
\end{definition}

\noindent If  $\varOmega = 0$, then $ 	\pounds_\xi R^h_{ijk} = 0$ and $\xi^{i}$ is said to follow a curvature symmetry  on $V_4$ or simply to write it generate a curvature collineation.
The CI equation \eqref{1.5N}	 can be written in component form as
\begin{equation}\label{1.6}
	R ^h_{ijk; l} \xi^l - R^l _{ijk} \xi^h_{;l} + R^h_{ljk} \xi ^l_{;i} +R^h_{ilk} \xi ^l_{;j} +R^h_{ijl} \xi ^l_{;k} = 2 \varOmega  R^h_{ijk}.
\end{equation}

\begin{definition}
	\cite{Duggal}	A spacetime $(V_4, ~g)$ admits Ricci inheritance (RI) symmetry along a 	 smooth vector field $\xi^{i}$,   if it satisfies the following equation 
	
	\begin{equation} \label{1.7N}
		\pounds_\xi R_{ij} = 2 \varOmega R_{ij}.
	\end{equation}
	
\end{definition}
\noindent Contraction of  \eqref{1.5N} gives  \eqref{1.7N}. Thus, in general, every curvature inheritance is  Ricci inheritance symmetry (i.e., CI $\Rightarrow$ RI), but the converse may not hold. Also,  RI reduces to RC when $\varOmega =0$ and  for $\varOmega \neq 0$, it is called as proper RI.

The literature on study of spacetimes    exhibits, a deep interest towards the  research of  different symmetries (in particular, curvature,  Ricci, projective,  matter, semi-conformal symmetry, conharmonic curvature inheritance). These geometrical symmetries are much helpful for obtaining exact solutions of Einstein field equations \eqref{1.2}. Such work of researches  \cite{klduggal2,KarmerD} motivates  us to inquire   the curvature  inheritance symmetry  in the spacetime,  admitting   M-projective curvature  tensor. The plan of this  paper is as follows:

After Section 1, the  introduction and preliminaries, we study curvature inheritance in an M-projectively flat spacetime. Section 3 deals with  M-projectively  flat perfect fluid   spacetime admitting curvature inheritance symmetry and examination of  some properties of such a spacetime. Equation of state is derived in Section 4 which is physically  significant  in cosmology. Finally,  we obtained some interesting results   considering a  purely electromagnetic distribution. In the last section, we provide an example in the form of a metric which is M-projectively flat.

\section{\bf Curvature Inheritance  in  M- Projectively Flat Spacetimes}

Consider an M-projectively flat Lorentzian manifold $(V_4, g)$, therefore \eqref{1.1} under condition   $\bar{W}^{h}_{ijk}$ = 0, implies
\begin{equation}\label{2.1}
	R^h_{ijk} 	= \frac{1}{6}( \delta^h_k R_{ij} - \delta^h_j R_{ik} + R^h_k g_{ij} - R^h_j g_{ik}).
\end{equation}

\noindent Contracting with respect to $h$ and $k$, we get
\begin{equation}\label{2.2}
	R_{ij} = \frac{R}{4} g_{ij}, 
\end{equation}	

\noindent here $ R$ is  scalar curvature. Thus,  we can state that  \textquotedblleft An M-projectively flat spacetime is an Einstein manifold.\textquotedblright~  From equations \eqref{2.1} and \eqref{2.2}, we have   
\begin{equation}\label{2.3}
	R^h_{ijk} 	= \frac{R}{12}(  g^h_k g_{ij} - g^h_j g_{ik}).
\end{equation}
Thus,  a result  by Zengin \cite{Zengin1}  is  mentioned  as  \textquotedblleft  An M-projectively flat spacetime is of constant curvature\textquotedblright.  Also,   M-projectively flat Lorentzian manifold is Ricci symmetric, i.e. $ \nabla_{k}  R_{ij} = 0$. 

For $(V_4,g)$ admitting Ricci inheritance,  \eqref{1.7N}   and  \eqref{2.2} lead to $\pounds_\xi g_{ij} = 2\varOmega g_{ij}.$ Thus, we have the following:

\begin{theorem}\label{4T2.1}
	Every Ricci inheritance in an M-projectively flat spacetime is a conformal motion.
\end{theorem}

Since,   every  M-projectively  flat spacetime is Einstein \cite{ZAhsan1},  thus from  Theorem \ref{4T2.1}, we can state the following:

\begin{corollary}\label{2.1.1}
	Every Ricci inheritance in an Einstein manifold is a conformal motion.
\end{corollary}

We know that a  space with non-zero constant curvature  or every harmonic space  is  Einstein. Then it yields the following:

\begin{corollary}
	Every Ricci inheritance in a space of constant curvature is a conformal motion.
\end{corollary}

\begin{corollary}\label{2.1.3}
	Every Ricci inheritance in a harmonic space is a conformal motion.
\end{corollary}

We know that in a semi-Riemannian space, every curvature inheritance is a Ricci inheritance \cite{Duggal}. Hence, we have

\begin{corollary}\label{2.1.4}
	Every curvature inheritance  in an M-projectively flat spacetime is a conformal motion with the conformal function $\varOmega$.
\end{corollary}

On similar lines of Corollaries (\eqref{2.1.1} - \eqref{2.1.3}), we simply write the following:

\begin{corollary}
	Every curvature inheritance in an Einstein spacetime is a conformal motion.
\end{corollary}

\begin{corollary}
	Every curvature inheritance in a space of constant curvature is a conformal motion.
\end{corollary}

\begin{corollary}
	Every curvature inheritance in a harmonic space is a conformal motion.
\end{corollary}

Now,  we consider the Lorentzian space of two dimensional.   Since, every $V_{2} $ is an Einstein space,  where scalar curvature $R$ is a non-zero function of coordinates, \eqref{1.7N} and \eqref{2.2} lead to (assuming that $V_{2} $ admits a Ricci inheritance) 
\begin{equation}
	4(2 \varOmega R_{ij}) = R (\pounds_\xi g_{ij}) + (\pounds_\xi R) g_{ij}.
\end{equation}
\noindent Or
\begin{equation}\label{2.5N}
	\pounds_\xi g_{ij} = \left[2\varOmega - \frac{\pounds_\xi R}{R}\right] g_{ij}, \
\end{equation}
$$ \implies \pounds_\xi g_{ij} = \left[2\varOmega -g\left(\frac{gradR}{R}, \xi\right)\right] g_{ij}$$	
In particular if $\varOmega=\frac{1}{2}\  g\left(\frac{gradR}{R}, \xi\right)$, then $\xi$ is a Killing vector. \\	
\noindent Hence, we have the following result

\begin{theorem}\label{Th2.2}
	Every Ricci inheritance in a 2-dimensional Riemannian space is a conformal motion.	
\end{theorem}

\begin{corollary}
	If in a $2$-dimensional Lorentzian space admitting Ricci inheritance, $grad(R)$ (i.e. the scalar potential) is orthogonal to $\xi$, then $\xi$ is a Killing vector. 
\end{corollary}

Since, CI $\Rightarrow$ RI, Theorem \ref{Th2.2} gives the following result.

\begin{corollary}
	Every curvature inheritance in a 2-dimensional Lorentzian space is a conformal motion.	
\end{corollary}

\noindent If the  scalar curvature $R$ holds the  inheritance symmetry  property in a 2-dimensional Lorentzian space, i.e., $\pounds_\xi R  = 2\varOmega R $ or $\pounds_\xi R -  2\varOmega R = 0$,  then equation  \eqref{2.5N} implies $\pounds_\xi g_{ij} = 0$. Thus, we have the following: 

\begin{corollary}
	If the scalar curvature of $V_2$ holds inheritance symmetry, then every curvature inheritance in  $V_{2}$ is  a motion.
\end{corollary}

\noindent Combining  the equations \eqref{2.2} and EFEs \eqref{1.2},   we get
\begin{equation}\label{2.5}
	T_{ij} = \frac{1}{\kappa} ( \wedge - \frac{R}{4}   ) g_{ij},\ \  \kappa \neq 0.
\end{equation}

\noindent Thus,  we have the following:
\begin{proposition}
	 An M- projectively flat spacetime,  the energy-momentum tensor holds EFEs with a cosmological term is in the form of equation \eqref{2.5}.
\end{proposition} 

\begin{theorem}
	If an  M- projective curvature tensor in $(V_4, g)$ admits curvature inheritance  along  a Killing vector field $ \xi$, then M- projective curvature tensor follows the symmetry inheritance property along $\xi$.
\end{theorem}

\begin{proof}
	The Lie differentiation of equation   \eqref{1.1} along to Killing vector field $\xi$ leads to 
	\begin{equation}\label{2.6}
		\pounds_\xi \bar{W}^{h}_{ijk}  = \pounds_\xi  R^h_{ijk} + \frac{1}{6} [\delta^h_j  \pounds_\xi R_{ik}-\delta^h_k \pounds_\xi  R_{ij}+ \pounds_\xi ( R^h_j ) g_{ik}-  \pounds_\xi ( R^h_k) g_{ij}] .
	\end{equation}
	If we consider  a spacetime $ (V_4, g)$ admits curvature inheritance symmetry  then from \eqref{2.6}, we  obtain

	\noindent ~~~~~~~~~~~~~~~~~~~~~~~~~~~~ $ \pounds_\xi \bar{W}^{h}_{ijk} = 2 \varOmega  [  ~ R^h_{ijk} + \frac{1}{6}(\delta^h_j R_{ik}-\delta^h_k R_{ij}+ R^h_j g_{ik}- R^h_k g_{ij}) ],$ \\	
	\noindent In view of \eqref{1.1}, the above equation entails
	\begin{equation}\label{2.7}
		\pounds_\xi \bar{W}^{h}_{ijk}   =  2 \varOmega   \bar{W}^{h}_{ijk}.
	\end{equation}
	
	\noindent This leads the proof.
\end{proof}

\begin{theorem}
	If an M- projective curvature tensor on a $( V_4, g)$ with a CKV $ \xi$ satisfies the curvature inheritance property, then M- projective curvature tensor follows the symmetry inheritance property along $\xi$.
\end{theorem}

\begin{proof}
	The Lie differentiation of equation   \eqref{1.1} along to Killing vector field $\xi$ leads to
	\begin{equation}\label{2.8}
		\pounds_\xi \bar{W}^{h}_{ijk}  = \pounds_\xi  R^h_{ijk} + \frac{1}{6} [\delta^h_j  \pounds_\xi R_{ik}-\delta^h_k \pounds_\xi  R_{ij}+ \pounds_\xi ( R^h_j ) g_{ik}-  \pounds_\xi ( R^h_k) g_{ij}  + 2 \varOmega  ( R^h_j  g_{ik} -  R^h_k g_{ij} )] .
	\end{equation}
	If a  spacetime $ (V_4, g)$ admits curvature inheritance symmetry  then from \eqref{2.8}, we  obtain  
	
	\noindent ~~~~~~~~~~~~~~~~~~~~~~~~~~~~ $ \pounds_\xi \bar{W}^{h}_{ijk} = 2 \varOmega  [  ~ R^h_{ijk} + \frac{1}{6}(\delta^h_j R_{ik}-\delta^h_k R_{ij}+ R^h_j g_{ik}- R^h_k g_{ij}) ].$ \\	
	\noindent By virtue of the above equation, we obtain
	
	\begin{equation}\label{2.9}
		\pounds_\xi \bar{W}^{h}_{ijk}   =  2 \varOmega   \bar{W}^{h}_{ijk}.
	\end{equation}
	
	\noindent This gives the proof.
\end{proof}

For M-projectively flat spacetimes we can simply obtain some results on curvature inheritance symmetry. Also, there are some results already in literature which can be proved for the M-projectively flat spacetime. For example, \textquotedblleft every proper CI in an Einstein space $( R \neq 0 ) $ is a proper conf M   with a conformal function $\varOmega $\textquotedblright  given by Duggal \cite{Duggal} is obvious in the M-projectively flat spacetime.

\section{\bf Curvature Inheritance  in   M- Projectively flat perfect fluid spacetimes}
In a perfect fluid spacetime, 
\begin{equation}\label{3.1}
	T_{ij}=(\mu + p) u_{i}u_{j}+ p g_{ij},
\end{equation} 
\noindent where the symbols are explained  after equation \eqref{1.3n}. 
On contraction of  \eqref{2.5}, we obtain
\begin{equation} \label{3.2}
	T = \frac{1}{ \kappa } (4 \wedge - R    ) ~~~~~~~~~~~~or~~~~~~~~~~~R = ( 4 \wedge - \kappa T ),
\end{equation}
and      contraction  of   equation \eqref{3.1}, leads to 
\begin{equation}\label{3.3}
	T =  (  3p  - \mu ).
\end{equation}   
Use of  equations \eqref{3.2} and \eqref{3.3}, yields 
\begin{equation}\label{3.4}
	R = [ 4\wedge + ( \mu - 3 p) \kappa].
\end{equation}
Next,   we use a  dynamic result for a perfect fluid  spacetime  along a conformal Killing vector field   $\xi^{i}$,  consider the equation \cite{Duggal}  
\begin{equation}\label{3.5}
	\pounds_\xi \mu = -2 \square\varOmega   -2 \varOmega  \mu -2 \varOmega _{;ij} u^{i} u^{j}, 
\end{equation}   
where $\square \varOmega = - \frac{1}{3} \varOmega R  $ and  $\varOmega_{;ij}=-\frac{1}{3}\varOmega R_{ij}$.  In a  perfect fluid spacetime, comparing with EFEs \eqref{1.2} and equation  \eqref{3.1} with the condition $ u_{i} u^{i}= -1$, we get
\begin{equation}\label{3.6}
	R_{ij} u^{i} u^{j} =( \kappa \mu-\frac{R}{2} + \wedge )= [ \kappa (\frac{3p+ \mu}{2}) - \wedge ].
\end{equation} 
If we set
\begin{equation}\label{3.7}
	\varOmega_{;ij} = \frac{\varOmega}{2} [ \frac{R}{3} g_{ij} -2 R_{ij}], 
\end{equation}
then from \cite{Duggal},   every  CIV  is also  a CKV in  an M- projectively flat spacetime.  If we  multiply  both sides by $ u^{i} u^{j}$ in \eqref{3.7}, using equation \eqref{3.6} and $u^{i}$ to be  timelike then we obtain
\begin{equation}\label{3.8}
	\varOmega _{;ij} u^{i} u^{j}= \varOmega  ( \frac{R}{3}- \kappa \mu - \wedge ) = -\frac{\varOmega }{3} [(2\mu+ 3p) \kappa - \wedge ].
\end{equation}
Now   using $\square \varOmega = - \frac{1}{3} \varOmega R  $  and \eqref{3.8} in  \eqref{3.5}, we get 
\begin{equation}\label{3.9}
	\pounds_\xi \mu =  2 \varOmega [ ( \kappa - 1) \mu +\wedge]. 
\end{equation}

\noindent Next,  we find   $\pounds_\xi p$ in M- projectively flat spacetime as follows:

\begin{equation}\label{3.10}
	\pounds_\xi p =  \frac{4}{3} \square~ \varOmega -2\varOmega  p- \frac{2}{3} \varOmega _{;ij} u^{i} u^{j}.
\end{equation}
By using $\square \varOmega = - \frac{1}{3} \varOmega R$ and  \eqref{3.8} in \eqref{3.10}, we get 
\begin{equation}\label{3.11}
	\pounds_\xi p =  2 \varOmega [ ( \kappa - 1)p -  \wedge]. 
\end{equation}
Thus, we have the following result
\begin{theorem}\label{2T2.5}
	If an M-projectively flat spacetime $( V_4, g )$ with  perfect fluid holds  EFEs with a cosmological constant admits the curvature inheritance symmetry along a vector field $\xi$, then we have the following
	\begin{equation*}
		(a)~~~	\pounds_\xi \mu =  2 \varOmega [ ( \kappa - 1) \mu + \wedge]	
	\end{equation*}
	\begin{equation*}
		(b)~~~	\pounds_\xi p =  2 \varOmega [ ( \kappa - 1)p -  \wedge].
	\end{equation*}
\end{theorem}
\noindent Consider the following special cases:

\noindent $\mathbf{Case ~1:}$  Under the hypothesis of Theorem \ref{2T2.5}, if   $( V_4, g )$ is  a perfect fluid  hold the EFEs  without  a cosmological term (i.e. $\wedge$ = 0), then we have the  following inheritance equations

\begin{equation}\label{3.12}
	(a)~~~	\pounds_\xi \mu =  2 \varOmega  ( \kappa - 1) \mu,
\end{equation}
\begin{equation}\label{3.13}
	(b)~~~	\pounds_\xi p =  2 \varOmega  ( \kappa - 1)p.
\end{equation}

\noindent $\mathbf{Case~2:}$ Under the hypothesis of Theorem \ref{2T2.5}, if   $( V_4, g )$ is a perfect fluid spacetime satisfying the EFEs without  a cosmological term (i.e. $\wedge$ = 0 and $\kappa = 1$ ), then we have the  following  equations

\begin{equation}
	(a)~~~	\pounds_\xi \mu =  0,
\end{equation}
\begin{equation}
	(b)~~~	\pounds_\xi p =   0.
\end{equation}

A spacetime $( V_4, g)$ inherits symmetry along a homothetic vector field $\xi$, if the following  equations hold for a perfect fluid:

$(a)~~~	\pounds_\xi \mu =  2\varOmega \mu$~~~	$(b)~~~	\pounds_\xi p =   2 \varOmega p$~~~  $(c)~~~	\pounds_\xi u^{i} =  2\varOmega u^{i}$.

\noindent Recently, the interest of many researcher has increased towards the study of the spacetime symmetry inheritance along a conformal Killing vector field $\xi$.  Coley and Tupper \cite{AA Coley} have worked on special CKV. Now, it will be demonstrated that how the  Theorem \ref{2T2.5}  is used to change the symmetry equations \eqref{3.12} and \eqref{3.13}  for a curvature inheritance vector field in an  M-projectively flat spacetime. Therefore,  consider  a conformal Killing vector field $\xi$,   such that from \cite{klduggal2} 

\begin{equation}\label{3.16}
	\pounds_\xi u^{i} =  2\varOmega u^{i} + v^{i}
\end{equation}
\noindent where $v^{i}$ is a spacelike vector orthogonal to $u^{i}$. Einstein field equations \eqref{1.2} for a perfect fluid, with \eqref{3.1}, imply that $u^{i}$  eigenvector of $R_{ij}$, and $u^{i}u_{i} = -1$
\begin{equation}\label{3.17}
	R_{ij}  u^{j} = [ - \kappa (\frac{3p+ \mu}{2}) + \wedge ]u_{i}.
\end{equation}
If $\xi$ is the curvature inheritance  vector or Ricci inheritance vector, then using  \eqref{3.17}  in \eqref{3.7}, we get 
\begin{equation}\label{3.18}
	\varOmega _{;ij} u^{i} u^{j} = -\frac{\varOmega }{3} [(2\mu+ 3p) \kappa -  \wedge ].
\end{equation}
Using the above equation in $ (\mu + p) v_{i}$ =  $2 [ \varOmega _{;ij} u^{j}   +	(\varOmega _{;kl} u^{k} u^{l}) u_{i} ]$, we find that $v_{i}
$ = 0 for $(\mu + p)$ $\neq$ 0. Then, equation \eqref{3.16} reduces to 

\begin{equation}\label{3.19}
	\pounds_\xi u^{i} =  2\varOmega u^{i}. 
\end{equation}	

\noindent The calculation of equations   \eqref{3.12}, \eqref{3.13} and \eqref{3.19} can be summarized in the following theorems.
\begin{theorem}
	If an M-projectively flat spacetime with a  perfect fluid obeying the Einstein field equations without a cosmological term admits a curvature inheritance symmetry along the vector field $\xi$, then symmetry with respect to $\xi$ is inherited by the spacetime. 
\end{theorem}

\begin{theorem}
	Let $(V_4, g)$ be an M-projectively flat spacetime with a perfect fluid obeying the Einstein field equation and admitting curvature inheritance symmetry along vector field $\xi$. If the perfect fluid denotes radiation era or stiff matter then this spacetime does not contain the cosmological constant.
\end{theorem}

\begin{proof}
	Let us assume that the perfect fluid denotes the radiation era, i.e., $\mu = 3p$. By taking the Lie derivative of this equation and using \eqref{3.9} and \eqref{3.11}, we can see that $\wedge$ = 0. Similarly, if the perfect fluid denotes stiff matter, i.e., $\mu = p$, then we also get $\wedge = 0$. This completes the proof.
\end{proof}

\begin{theorem}
	Let $(V_4, g)$ be an M-projectively flat spacetime with a perfect fluid obeying the Einstein field equations and admitting curvature inheritance symmetry along vector field $\xi$. If the Lie derivative of   $ (\mu + p )$ is zero, then either the matter content of the spacetime satisfying  the vacuum-like equation of state or  $ (\mu + p )$ is constant. 
\end{theorem}

\begin{proof}
	By using \eqref{3.9} and \eqref{3.11}, we get 
	\begin{equation}
		\pounds_\xi (\mu + p) = \pounds_\xi \mu + \pounds_\xi p = 0.
	\end{equation}
	Finally, we can say that either the  matter content of the spacetime satisfy the vacuum like equation of state or  $ (\mu + p )$ is constant.   Thus, the   proof is completed.
\end{proof}

\section{\bf Equation of State}
Let $(V_4, g)$ be  an  M-projectively flat spacetime admitting curvature inheritance  along a vector field $\xi$. From Corollary \ref{2.1.4}, $\xi$  is a CKV,  satisfying the following equation
\begin{equation}\label{4.1}
	(\xi^{j} R_{ij})_{;i} = -3 \square \psi.
\end{equation}
Using Einstein field equations,  \eqref{4.1} leads to

\begin{equation}\label{4.2}
	([k T^{ij}  - (\wedge -\frac{R}{2}) g^{ij}  ]\xi_{j}) _{;i} = -3 \square \psi.
\end{equation}
As  $-3 \square \psi = \alpha R$, therefore  
\begin{equation}\label{4.3}
	([k T^{ij}  - (\wedge -\frac{R}{2}) g^{ij}  ]\xi_{j}) _{;i} = \alpha R.
\end{equation}
Further, for the perfect fluid spacetime with  \eqref{3.1}, we have 
\begin{equation}\label{4.5}
	[k p -\wedge + \frac{R}{2} ]\xi^{i} _{;i} = \alpha R.
\end{equation}
Equation \eqref{4.5},  on substitution   $\xi^{i}_{;i}$ = $4\alpha$ reduces to
\begin{equation}\label{4.6}
	4 (k p - \wedge) = - R.
\end{equation}
Here $\xi$ is CKV. For an  M-projectively flat spacetime,  the  scalar curvature  $ R = 4\wedge - k (3p - \mu) $. Now comparing  with \eqref{4.6}, we have 
\begin{equation}\label{4.7}
	\mu + p = 0,  ~~~~(k\neq 0)
\end{equation}

\noindent which may describe two cases (i) empty spacetime $\mu$ = $p$ =0  and (ii)  perfect fluid spacetime  holding  vacuum like equation of state.

\noindent Thus,  we state the  result as follows

\begin{theorem}
	
	Let an    M-projectively flat spacetime with perfect fluid obeying the EFEs with cosmological constant be admit  curvature inheritance symmetry along a  vector field $\xi$. Then $(V_4, g)$  is either a vacuum or satisfies the vacuum-like equation of state.	 
\end{theorem}

From \cite{KarmerD}, we note that    $\mu +  p = 0 $ implies that scalar curvature is equal to cosmological constant   and Ojha \cite{RH1}  termed  this  as {\it Phantom~ Barrier}.  Alan Guth in 1981 proposed the idea of {\it cosmic inflation} \cite{Guth}, and explained the similar conditions for the same in the universe. Amendola and Tsujikawa \cite{LS} explained the term inflation  in their  paper   as rapid expansion of the spacetime that might occurred just after the {\it Big~ Bang}. In the light of above discussion, we have

\begin{theorem}
	Let a perfect fluid M-projectively flat spacetime  $(V_4, g)$ obeying the EFEs with cosmological constant, admits curvature inheritance symmetry along a  vector field $\xi$. Then $(V_4, g)$   represents inflation.
\end{theorem}

\noindent From equations \eqref{3.1} and \eqref{4.7}, we obtain 
\begin{equation}\label{4.8}
	T_{ij} = p g_{ij}.
\end{equation}
Since, an M-projectively flat spacetime is Einstein, thus it is  of constant scalar curvature. Further as  $\wedge$ and $k$ are  constants, therefore  equation  \eqref{4.6} confirms that $p$ is constant. Hence $ \mu  = - p $ is also constant.  This condition has a special
significance in the physics of spacetime. With this condition, the fluid starts behaving like a cosmological constant \cite{KarmerD}, which is also called Phantom
Barrier \cite{SNR}. This causes inflation of the universe \cite{LS}.

\noindent Thus, we obtain the following theorem:
\begin{theorem}
	The isotropic pressure and energy density of an M-projectively flat perfect fluid spacetime satisfying EFEs with cosmological constant  admitting the curvature inheritance symmetry along a vector field $\xi$ are constants. 
\end{theorem}

\section{\bf A Purely Electromagnetic Distribution Admitting Curvature Inheritance Symmetry}
Section $5$  is devoted to the study of   few results for  purely electromagnetic distribution. Contraction of \eqref{4.8},  provides
\begin{equation}\label{5.1}
	T = 4 p.
\end{equation}

\noindent Using the equations \eqref{5.1} and \eqref{3.3}, we get  $p$ = $-\mu$ = $\frac{T}{4}.$ Now since for a purely electromagnetic distribution, $T = 0$.  We have  $p$ = $\mu$ = 0, in other words  \textquotedblleft In a purely electromagnetic distribution, the isotropic pressure and the energy density of a perfect fluid spacetime satisfying Einstein field equations with cosmological constant  vanish."\\ 

\noindent Putting $\mu$ = $p$ = 0 in \eqref{3.1},  we get $T_{ij}$ = 0.  This implies that the spacetime is vacuum. Taking contraction of \eqref{1.2}, we get
\begin{equation}\label{5.6}
	R = 4 \wedge - k T,
\end{equation}
substituting  \eqref{5.6} in EFEs  \eqref{1.2} and   for purely electromagnetic distribution, we get
\begin{equation}\label{5.7}
	R_{ij} = \wedge~ g_{ij}.
\end{equation}
\noindent If  M-projectively flat    perfect fluid spacetime  $(V_4, g)$    admits curvature inheritance symmetry along $\xi$. Then Lie derivative of  equation \eqref{5.7}  leads  to a conformal motion.  Thus,  we may  state  the following result:
\begin{theorem}
Th	curvature inheritance symmetry in  an M-projectively flat perfect fluid spacetime for a purely electromagnetic distribution, which satisfies   EFEs with  a cosmological  term is  conformal motion.
\end{theorem}

\noindent In this case, we choose Einstein  field equations without cosmological constant takes the form
\begin{equation}\label{5.2}
	R_{ij} - \frac{R}{2} g_{ij} =  k ~  T_{ij}.
\end{equation}
\noindent Contracting \eqref{5.2} we get
\begin{equation}\label{5.3}
	R = - k T. 
\end{equation}
\noindent For purely electromagnetic distribution $T$ = 0,  \eqref{5.3} reduces to $R$ = 0. Therefore, from \eqref{2.3},  we obtain 
\begin{equation}
	R^h_{ijk} = 0.
\end{equation}
This implies that $(V_{4}, g)$ is an Euclidean space. Thus, we have  the following result:
\begin{corollary}
	In a purely electromagnetic distribution, an M-projectively
	flat perfect fluid spacetime  $(V_4, g)$  satisfying   EFEs without a cosmological constant does not admit any curvature symmetry.\\
\end{corollary}

\section{\bf Example of an M-Projectively Flat Spacetime}

Let $\mathbb{R}^4$ be equipped with the following metrics in terms of coordinates $x^1$, $x^2$, $x^3$ and $x^4$ given  as follows:
\begin{equation}\label{6.1} 
	ds^2=-(dx^{1})^{2}+e^{x^1}[(dx^{2})^{2}+(dx^{3})^{2}+(dx^{4})^{2}],
\end{equation}
\begin{equation}\label{6.2} 
	ds^2=(dx^{1})^{2}+e^{x^1}[(dx^{2})^{2}+(dx^{3})^{2}+(dx^{4})^{2}].
\end{equation}
We note that the metric in (\ref{6.1}) (resp. (\ref{6.2})) is a Lorentzian (resp. Riemannian) metric.
\index The non-vanishing  components of second kind   Christoffel symbols   ($\Gamma^h_{ij}$) of the metric \eqref{6.1} are given by  
$$\begin{array}{c}
	
	\Gamma^2_{12}=\Gamma^3_{13}=\Gamma^4_{14}=\frac{1}{2}, \
	\Gamma^1_{22}=\Gamma^1_{33}=\Gamma^1_{44}=\frac{e^{x^1}}{2}. 	
\end{array}$$

\indent The non-vanishing  components of the Ricci tensor $R_{ij}$ and   Riemann-Christoffel curvature tensor $R_{hijk}$ of the  metric \eqref{6.1} are given by 

$$\begin{array}{c}
	R_{1212}=R_{1313}=R_{1414}=-\frac{e^{x^1}}{4}, \
	R_{2323}=R_{2424}=R_{3434}=\frac{e^{2x^1}}{4},
\end{array}$$

$$\begin{array}{c}
	R_{11}=\frac{3}{4},\
	
	R_{22}=R_{33}=R_{44}= -\frac{3e^{x^1}}{4}. \\
\end{array}$$

Also, for the metric \eqref{6.1}, the scalar curvature is   $R=-3$. It is easy to compute that all the components of the M-projective curvature tensor are zero. Thus we have the  following result:

\begin{theorem}
	Let $\mathbb R^4$ be equipped with  the Lorentzian metric given in \eqref{6.1}, then $(\mathbb R^4,g)$ is an M-projectively flat spacetime.
\end{theorem} 

\begin{remark}
	We note that $(\mathbb R^4,g)$ endowed with the Riemannian metric \eqref{6.2} is also M-projectively flat.
\end{remark}

\begin{remark}
	We mention that the metrics \eqref{6.1} and \eqref{6.2} are warped product $M_1\times_f M_2$, where the warping function $f=e^{x^1}$ with $1$-dimensional base $M_1$ and $3$-dimensional fibre $M_2$.
\end{remark}

\end{document}